\documentclass[12pt]{amsart}  %document style grand article journal
\usepackage{graphicx}%g�re les graphiques
\usepackage{here} %place les photos o� il faut
\usepackage{array} %faire des tableaux plus compliqu�s

\usepackage{vmargin}
\usepackage{amssymb}
\usepackage{mathrsfs}
\usepackage[all]{xy}
\usepackage[usenames,dvipsnames]{color}%64 couleurs
\usepackage{soul}
\RequirePackage[colorlinks,linkcolor=blue,citecolor=LimeGreen,urlcolor=red]{hyperref} %ref actives en couleur au lieu de carres
\usepackage{amsmath}
\newtheorem{theorem}{Theorem}[section]

\newtheorem{lemma}[theorem]{Lemma}

\newtheorem{prop}[theorem]{Proposition}
\newtheorem{remark}[theorem]{Remark}

\numberwithin{equation}{section}

\newcommand{\R}{\mathbb{R}}

\newcommand{\T}{\mathbb{T}}
\renewcommand{\S}{{\mathbb S}}

\newcommand{\func}[3]{#1 : #2 \longrightarrow #3}

\newcommand{\disp}{\displaystyle}
\newcommand{\abs}[1]{\left|#1\right|}
\newcommand{\eps}{\varepsilon}
\newcommand{\norm}[1]{\left\|#1\right\|}

\renewcommand{\leq}{\leqslant}
\renewcommand{\geq}{\geqslant}

\newcommand{\pa}[1]{\left(#1\right)}
\newcommand{\cro}[1]{\left[#1\right]}
\newcommand{\br}[1]{\left\{#1\right\}}
\newcommand\restr[2]{{% we make the whole thing an ordinary symbol
  \left.\kern-\nulldelimiterspace % automatically resize the bar with \right
  #1 % the function
  \right|_{ #2} % this is the delimiter
  }}

\def\signmb{\bigskip \begin{center} {\sc
Marc Briant\par\vspace{3mm}
Sorbonne Universit\'es, UPMC Univ. Paris 06/ CNRS\par
UMR 7598, Laboratoire Jacques-Louis Lions,\par
F-75005, Paris, France\par
\vspace{3mm}
e-mail:} \tt{briant.maths@gmail.com} \end{center}}

\begin{document} 

\title[Stability for the multi-species Boltzmann equation in $L^\infty$ settings]{Stability of global equilibrium for the multi-species Boltzmann equation in $L^\infty$ settings}
\author{Marc Briant}
%\thanks{}
%\thanks{}

\begin{abstract}
We prove the stability of global equilibrium in a multi-species mixture, where the different species can have different masses, on the $3$-dimensional torus. We establish stability estimates in $L^\infty_{x,v}(w)$ where $w=w(v)$ is either polynomial or exponential, with explicit threshold. Along the way we extend recent estimates and stability results for the mono-species Boltzmann operator not only to the multi-species case but also to more general hard potential and Maxwellian kernels.

\end{abstract}

\maketitle

\vspace*{10mm}
%\smallskip


%%\smallskip
%%\textbf{AMS Subject Classification}: 82C40 Kinetic theory of gases,
%%76P05 Rarefied gas flows, Boltzmann equation, 54C70 Entropy, 60J75
%%Jump processes.
%

\tableofcontents

\section{Introduction} \label{sec:intro}

The multi-species Boltzmann equation rules the dynamics of a dilute gas composed of $N$ different species of chemically non-reacting mono-atomic particles. More precisely, this equation describes the time evolution of $F_i(t, x, v)$, the distribution of particles of the $i^{th}$ species in position and velocity, starting from an initial distribution. It can be modeled by the following system of Boltzmann equations, stated on $\R^+\times\T^3\times\R^3$,

\begin{equation}\label{multiBE}
\forall\: 1\leq i \leq N, \quad \partial_tF_i(t,x,v) + v\cdot \nabla_x F_i(t,x,v) = Q_i(\mathbf{F})(t,x,v)
\end{equation}
with initial data
$$\forall\: 1\leq i \leq N,\:\forall (x,v)\in \T^3\times\R^3, \quad F_i(0,x,v) = F_{0,i}(x,v).$$
Note that the distribution function of the system is given by the vector $\mathbf{F} = (F_1,\dots,F_N)$.

\bigskip
The Boltzmann operator $Q(\mathbf{F})=(Q_1(\mathbf{F}),\ldots, Q_N(\mathbf{F}))$ is given for all $i$ by 
$$Q_i(\mathbf{F}) = \sum\limits_{j=1}^N Q_{ij}(F_i,F_j),$$
where $Q_{ij}$ describes interactions between particles of either the same ($i=j$) or of different ($i\neq j$) species and is local in time and space.
$$Q_{ij}(F_i,F_j)(v) =\int_{\R^3\times \mathbb{S}^{2}}B_{ij}\left(|v - v_*|,\mbox{cos}\:\theta\right)\left[F_i'F_j^{'*} - F_iF_j^*\right]dv_*d\sigma,$$
where we used the shorthands $F_i'=F_i(v')$, $F_i=F_i(v)$, $F_j^{'*}=F_j(v'_*)$ and $F_j^*=F_j(v_*)$ with the definition 
$$\left\{ \begin{array}{rl} \displaystyle{v'} & \displaystyle{=\frac{1}{m_i+m_j}\pa{m_iv+m_jv_* +  m_j|v-v_*|\sigma}} \vspace{2mm} \\ \vspace{2mm} \displaystyle{v' _*}&\displaystyle{=\frac{1}{m_i+m_j}\pa{m_iv+m_jv_* -m_i  |v-v_*|\sigma}} \end{array}\right., \: \mbox{and} \quad \mbox{cos}\:\theta = \left\langle \frac{v-v_*}{\abs{v-v_*}},\sigma\right\rangle .$$
These expressions are a way to express the fact that collisions happening inside the gas are only binary and elastic. Physically, it means that $v'$ and $v'_*$ are the velocities of two molecules of species $i$ and $j$ before collision giving post-collisional velocities $v$ and $v_*$ respectively, with conservation of momentum and kinetic energy:
\begin{equation}\label{elasticcollision}
\begin{split}
m_iv + m_jv_* &= m_iv' + m_jv'_*,
\\\frac{1}{2}m_i\abs{v}^2 + \frac{1}{2}m_j\abs{v_*}^2 &= \frac{1}{2}m_i\abs{v'}^2 + \frac{1}{2}m_j\abs{v'_*}^2.
\end{split}
\end{equation}
\par The collision kernels $B_{ij}$ encode the physics of the interaction between two particles. We mention at this point that one can derive this type of equations from Newtonian mechanics at least formally in the case of single species \cite{Ce}\cite{CIP}. The rigorous validity of the mono-species Boltzmann equation from Newtonian laws is known for short times (Landford's theorem \cite{La} and, more recently, \cite{GST}\cite{PSS}).  
\bigskip

%%%%%%%%%%%%%%%%%%%%%%%%%%%%%%%%%%%%%%%%%%%%%%%%%%%%%%%%%%%%%%%%%%%%%%%%%%%%%%%%%%%%%%%%%%%%%%%%%%%%%%%%%%%%%%%%%%%%%%%%%%%%%%%
%%%%%%%%%%%%%%%%%%%%%%%%%%%%%%%%%%%%%%%%%%%%%%%%%%%%%%%%%%%%%%%%%%%%%%%%%%%%%%%%%%%%%%%%%%%%%%%%%%%%%%%%%%%%%%%%%%%%%%%%%%%%%%%
%%%%%%%%%%%%%%%%%%%%%%%%%%%%%%%%%%%%%%%%%%%%%%%%%%%%%%%%%%%%%%%%%%%%%%%%%%%%%%%%%%%%%%%%%%%%%%%%%%%%%%%%%%%%%%%%%%%%%%%%%%%%%%%

\subsection{Global equilibrium, stability and the perturbative regime}

It is now well known \cite{DMS}\cite{DJMZ}\cite{BGS} that the symmetries of the collision operator imply the conservation of the total number density $c_{\infty,i}$ of each species, of the total momentum of the gas $\rho_\infty u_\infty$ and its total energy $3\rho_\infty\theta_\infty /2$:
\begin{equation}\label{conservationlaws}
\begin{split}
\forall t\geq 0,\quad& c_{\infty,i} = \int_{\T^3\times\R^3} F_i(t,x,v)\:dxdv \quad (1\leq i \leq N)
\\& u_{\infty} = \frac{1}{\rho_\infty}\sum\limits_{i=1}^N\int_{\T^3\times\R^3} m_ivF_i(t,x,v)\:dxdv
\\&\theta_{\infty} = \frac{1}{3\rho_\infty}\sum\limits_{i=1}^N\int_{\T^3\times\R^3} m_i\abs{v-u_\infty}^2F_i(t,x,v)\:dxdv,
\end{split}
\end{equation}
where $\rho_\infty = \sum_{i=1}^Nm_ic_{\infty,i}$ is the global density of the gas. These expressions already show that there exist non-trivial interactions between each species and the mixture, unlike independent single-species Boltzmann equations that preserve the momentum and energy of each of the species independently.
\par The operator $\mathbf{Q}=(Q_1,\dots,Q_N)$ also satisfies a multi-species version of the classical H-theorem \cite{DMS} from which one deduces that there exists a unique global equilibrium, i.e. a stationary solution $\mathbf{F}$ to $\eqref{multiBE}$, associated to the initial data $\mathbf{F_0}(x,v) =(F_{0,1},\dots,F_{0,N})$. It is given by the global Maxwellian 
$$\forall\:1\leq i \leq N,\quad F_i(t,x,v) = F_i(v)= c_{\infty,i}\pa{\frac{m_i}{2\pi k_B \theta_{\infty}}}^{3/2}\mbox{exp}\cro{-m_i\frac{\abs{v-u_{\infty}}^2}{2k_B\theta_{\infty}}}.$$
By translating and rescaling the coordinate system we can always assume that $u_\infty=0$ and $k_B\theta_\infty=1$ so that the only global equilibrium is the normalized Maxwellian
\begin{equation}\label{mui}
\boldsymbol\mu =\pa{\mu_i}_{1\leq i \leq N} \quad\mbox{with}\quad \mu_i(v) = c_{\infty,i}\pa{\frac{m_i}{2\pi}}^{3/2}e^{-m_i\frac{\abs{v}^2}{2}}.
\end{equation}

\bigskip
Recently, E. Daus and the author \cite{BriDau} solved the existence and uniqueness problem for the multi-species Boltzmann equation $\eqref{multiBE}$ around the global equilibrium $\boldsymbol\mu$ in $L^1_vL^\infty_x(1+|v|^k)$ for $k$ larger than an explicit threshold. They also proved the stability of $\boldsymbol \mu$ by showing the exponential decay in time of $\mathbf{F}(t)-\boldsymbol\mu$ as long as one starts sufficently close to the global equilibrium. 
\par The present work is intended to be a companion paper to \cite{BriDau} and we thus refer to it and the references therein for more details about previous works. Our aim is to obtain a similar perturbative theory but in $L^\infty_{x,v}(w)$, where the weight $w$ is either polynomial or stretched exponential, and thus show the exponential stability of $\boldsymbol \mu$. The importance of polynomial weights relies on the physically relevant problems of initial data having solely finite moments, rather than having a Maxwellian decay. Such results are very recent in the case of mono-species Boltzmann equation \cite{GMM} and the present paper fills up the gap for the multi-species case. The interest of this $L^\infty_{x,v}$ study is that, combined to the already mentionned $L^1_vL^\infty_x$ one, a mere interpolation arguments then offers the stability of $\boldsymbol\mu$ in every underlying Lebesgue spaces in the $v$ variable, with explicit threshold on the polynomial weight.
\par  More precisely, we study the existence, uniqueness and exponential decay of solutions of the form $F_i(t,x,v) = \mu_i(v) + f_i(t,x,v)$ for all $i$ when one starts close to $\boldsymbol \mu$. This is equivalent to solving the perturbed multi-species Boltzmann system of equations
\begin{equation}\label{perturbedmultiBE}
\partial_t \mathbf{f} + v\cdot\nabla_x\mathbf{f} = \mathbf{L}(\mathbf{f}) + \mathbf{Q}(\mathbf{f}), 
\end{equation}
or equivalently in the non-vectorial form
$$\forall \:1\leq i\leq N,\quad \partial_t f_i + v\cdot\nabla_x f_i = L_i(\mathbf{f}) + Q_i(\mathbf{f}),$$
where $\mathbf{f}=(f_1,\dots,f_N)$ and the operator $\mathbf{L} =(L_1,\dots,L_N)$ is the linear Boltzmann operator given for all $1\leq i \leq N$ by
$$L_i(\mathbf{f}) = \sum\limits_{j=1}^N L_{ij}(f_i,f_j) \quad\mbox{with}\quad L_{ij}(f_i,f_j) = Q_{ij}(\mu_i,f_j)+Q_{ij}(f_i,\mu_j).$$
The conservations of individual mass, total momentum and total energy $\eqref{conservationlaws}$ of $\mathbf{F}$ are translated onto $\mathbf{f}$ as follows 
\begin{equation}\label{perturbedconservationlaws}
\begin{split}
\forall t\geq 0,\quad& 0 = \int_{\T^3\times\R^3} f_i(t,x,v)\:dxdv \quad (1\leq i \leq N)
\\& 0 = \sum\limits_{i=1}^N\int_{\T^3\times\R^3} \cro{\begin{array}{c} m_iv \\ m_i\abs{v}^2\end{array}}f_i(t,x,v)\:dxdv.
\end{split}
\end{equation}
\par If the mono-species Boltzmann equation has been extensively studied in the perturbative context (we refer to the discussion in \cite{BriDau} and the exhaustive review \cite{UkYa}), the only result in our knowledge of a Cauchy theory for the multi-species case is \cite{BriDau}. We however mention the existing works \cite{BGPS}\cite{BGS}, where they studied the diffusive limit of the linear part of $\eqref{perturbedmultiBE}$, \cite{DJMZ}, where they obtained an explicit spectral gap for the linear operator $\mathbf{L}$, and finally \cite{DMS} where they derived the multi-species H-theorem and dealt with the case of chemically reacting species.
\par The main strategy of the present work is an analytic and non-linear adaptation of a recent extension result for semigroups \cite{GMM} that has also been used in \cite{BriDau} in a different setting. In a nutshell, we decompose the linear operator as $\mathbf{L} = -\boldsymbol\nu + \mathbf{A}+\mathbf{B}$ where $\boldsymbol\nu$ is a positive multiplication operator, $\mathbf{A}$ has some regularising properties and $\mathbf{B}$ acts like a ``small perturbation'' of $\boldsymbol\nu$ and decompose our full equation $\eqref{perturbedmultiBE}$ into a system of differential equations for $\mathbf{f_1}+\mathbf{f_2}=\mathbf{f}$
\begin{eqnarray*}
&&\partial_t \mathbf{f_1} +v\cdot\nabla_x\mathbf{f_1} = -\boldsymbol\nu(v)\mathbf{f_1} + \mathbf{B}\pa{\mathbf{f_1}}+ \mathbf{Q}(\mathbf{f_1}+\mathbf{f_2})
\\&&\partial_t \mathbf{f_2}+v\cdot\nabla_x\mathbf{f_2} = \mathbf{L}(\mathbf{f_2}) + \mathbf{A}(\mathbf{f_1}).
\end{eqnarray*}
The key contribution of this article is a generalisation of the control of the operator $\mathbf{B}$ in weighted $L^\infty_{x,v}$. This operator was estimated in \cite{GMM} in the case of mono-species with hard spheres. We extend the result not only to multi-species but above all to hard potential and Maxwellian kernels (see rigorous definition below). The possibility of having different masses is an intricate computational extension. Indeed, of important note from \eqref{mui} is that each species evolves, at equilibrium, at its own exponential rate and one has to understand how the linear operator actually mix these different speeds. Moreover, the non hard spheres case brings new difficulties and new behaviours for small relative velocities.
\par We conclude by emphasizing that our result includes the case of mono-species Boltzmann equation recently obtained \cite{GMM} and extend them to more general kernels. Our proofs will also involve to track down thoroughly explicit constants in order to exhibit the explicit threshold for the polynomial weights.
\bigskip

%%%%%%%%%%%%%%%%%%%%%%%%%%%%%%%%%%%%%%%%%%%%%%%%%%%%%%%%%%%%%%%%%%%%%%%%%%%%%%%%%%%%%%%%%%%%%%%%%%%%%%%%%%%%%%%%%%%%%%%%%%%%%%%
%%%%%%%%%%%%%%%%%%%%%%%%%%%%%%%%%%%%%%%%%%%%%%%%%%%%%%%%%%%%%%%%%%%%%%%%%%%%%%%%%%%%%%%%%%%%%%%%%%%%%%%%%%%%%%%%%%%%%%%%%%%%%%%
%%%%%%%%%%%%%%%%%%%%%%%%%%%%%%%%%%%%%%%%%%%%%%%%%%%%%%%%%%%%%%%%%%%%%%%%%%%%%%%%%%%%%%%%%%%%%%%%%%%%%%%%%%%%%%%%%%%%%%%%%%%%%%%

\subsection{Main result and organisation of the paper}
We start with some conventions and notations.
\par First, to avoid any confusion, vectors and vector-valued operators in $\R^N$ will be denoted by a bold symbol, whereas their components by the same indexed symbol. For instance, $\mathbf{W}$ represents the vector or vector-valued operator $(W_1,\dots,W_N)$.  We shall use the following shorthand notation
$$\langle v \rangle = \sqrt{1+\abs{v}^2}.$$
The convention we choose for functional spaces is to index the space by the name of the concerned variable, so we have for $p$ in $[1,+\infty]$
$$L^p_{[0,T]} = L^p\pa{[0,T]},\quad L^p_{t} = L^p \left(\R^+\right),\quad L^p_x = L^p\left(\T^3\right), \quad L^p_v = L^p\left(\R^3\right).$$
At last, for $\func{\mathbf{W}=(W_1, \ldots, W_N)}{\R^3}{\R^+}$ a strictly positive measurable function in $v$, we will use the following vector-valued weighted Lebesgue spaces defined by their norms
$$\norm{\mathbf{f}}_{L^{\infty}_{x,v}\pa{\mathbf{W}}} = \sum\limits_{i=1}^N \norm{f_i}_{L^{\infty}_{x,v}(W_i)} \quad\mbox{where}\quad   \norm{f_i}_{L^{\infty}_{x,v}\pa{W_i}} = \sup\limits_{(x,v)\in \T^3 \times \R^3}\big(\left|f_i(x,v)\right|W_i(v)\big).$$

\bigskip
We will use the following assumptions on the collision kernels $B_{ij}$. Note that there were the assumptions also made in \cite{BriDau} (less restrictive than \cite{DJMZ}) and translate into the commonly used description for the mono-species collision kernel \cite{Ce}\cite{CIP}.
\renewcommand{\labelenumi}{(H\theenumi)}
\begin{enumerate}
\item The following symmetry holds
$$B_{ij}(|v-v_*|,\cos\theta) = B_{ji}(|v-v_*|,\cos\theta)\quad\mbox{for }1\le i,j\le N$$
which means that the probability of a particle of species $i$ colliding with a particle of species $j$ is the same as $j$ colliding with $i$.
\item The collision kernels decompose into the product
$$ B_{ij}(|v-v_*|,\cos\theta) = \Phi_{ij}(|v-v_*|)b_{ij}(\cos\theta),
\quad 1\le i,j\le N,$$
where the functions $\Phi_{ij}\ge 0$ are called kinetic part and $b_{ij}\ge 0$ angular part. This is a common assumption as it is technically more convenient and also covers a wide range of physical applications.
\item The kinetic part has the form of hard or Maxwellian ($\gamma=0$) potentials, \textit{i.e.}
$$\Phi_{ij}(|v-v_*|)=C_{ij}^{\Phi}|v-v_*|^{\gamma}, \quad C_{ij}^{\Phi}>0,~\:\gamma\in[0,1], \quad \forall\: 1 \leq i,j\leq N$$
which describes inverse-power laws potential in between particles \cite[Section 1.4]{Vi2}.
\item For the angular part, we assume a strong form of Grad's angular cutoff (first introduced in \cite{Gr1}), that is: there exist constants $C_{b1}$, $C_{b2}>0$ such that
for all $1\le i,j\le N$ and $\theta\in[0,\pi]$,
$$  0<b_{ij}(\cos\theta)\le C_{b1}|\sin\theta|\,|\cos\theta|, \quad b'_{ij}(\cos\theta)\le C_{b2}.$$
Furthermore, 
$$  C^b := \min_{1\le i\le N}\inf_{\sigma_1,\sigma_2\in\S^2}\int_{\S^2}\min\big\{	b_{ii}(\sigma_1\cdot\sigma_3),b_{ii}(\sigma_2\cdot\sigma_3)\big\}\:d\sigma_3 > 0. $$
Again, this positivity assumption is satisfied by most of the physically relevant cases with Grad's angular cutoff, for instance hard spheres ($b=\gamma=1$). It is required even for the mono-species case in order to construct an explicit spectral gap for the linear Boltzmann operator \cite{DJMZ}\cite{BM}\cite{Mo1}.
\end{enumerate}

\noindent We emphasize here that the important cases of Maxwellian molecules ($\gamma=0$ and $b=1$) and of hard spheres ($\gamma=b=1$) are included in our study. We shall use the standard shorthand notations
\begin{equation}\label{constantsbij}
b_{ij}^\infty = \norm{b_{ij}}_{L^\infty_{[-1,1]}} \quad\mbox{and}\quad l_{b_{ij}} = \norm{b\circ \cos}_{L^1_{\S^2}}.
\end{equation}

\bigskip
Under these assumptions we shall prove the following theorem.
\bigskip
\begin{theorem}\label{theo:main}
Let the collision kernels $B_{ij}$ satisfy assumptions $(H1) - (H4)$ and let $w_i=e^{\kappa_1(\sqrt{m_i}\abs{v})^{\kappa_2}}$ with $\kappa_1 >0$ and $\kappa_2$ in $(0,2)$ or $w_i=\langle \sqrt{m_i}v \rangle^k$ with $k> k_0$ where $k_0$ is the minimal integer such that
\begin{equation}\label{k0}
C_B(\mathbf{w}) = \frac{4\pi}{k-1-\gamma}\max\limits_{1\leq i \leq N}\cro{\pa{\sum\limits_{j=1}^NC^\Phi_{ij}b_{ij}^\infty\frac{(m_i+m_j)^2}{m_i^{2-\frac{\gamma}{2}}m_j^{\frac{5+\gamma}{2}}}}\pa{\sum\limits_{1\leq k \leq N}\frac{\sqrt{m_k}}{C^\Phi_{ik}l_{b_{ik}}}}} <1.
\end{equation}
Then there exist $\eta_w$, $C_w$ and $\lambda_w >0$ such that for any $\mathbf{F_0} = \boldsymbol\mu + \mathbf{f_0} \geq 0$ satisfying the conservation of mass, momentum and energy $\eqref{conservationlaws}$ with $u_\infty=0$ and $\theta_\infty=1$, if
$$\norm{\mathbf{F_0} - \boldsymbol\mu}_{L^\infty_{x,v}(\mathbf{w})}\leq \eta_w$$
then there exists a unique solution $\mathbf{F}=\boldsymbol\mu + \mathbf{f}$ in $L^\infty_{x,v}(\mathbf{w})$ to the multi-species Boltzmann equation $\eqref{multiBE}$ with initial data $\mathbf{f_0}$. Moreover, $\mathbf{F}$ is non-negative, satisfies the conservation laws and
$$\forall t\geq 0, \quad \norm{\mathbf{F}-\boldsymbol\mu}_{L^\infty_{x,v}(\mathbf{w})} \leq C_w e^{-\lambda_w t}\norm{\mathbf{F_0}-\boldsymbol\mu}_{L^\infty_{x,v}(\mathbf{w})}.$$
The constants are explicit and only depend on $N$, $\mathbf{w}$, the different masses $m_i$ and the collision kernels.
\end{theorem}
\bigskip

\begin{remark}\label{rem:mainresults}
We make a few comments about the theorem above.
\begin{enumerate}
\item[(1)] As mentioned in the introduction the above Theorem has been proved in $L^1_vL^\infty_x(1+\abs{v}^k) \supset L^\infty_{x,v}(\mathbf{w})$ in \cite{BriDau}. We therefore do not need to tackle the issue of existence or uniqueness and are only left with proving the \textit{a priori} stability of $\boldsymbol\mu$ in the considered spaces.
\item[(2)] Unlike the classical mono-species Boltzmann equation, the natural weights for the multi-species case strongly depend on the mass of each species. Indeed,as given by $w_i$, each species has its own specific weight depending on its mass. As we shall see, this is needed to balance the cross-interactions generated inside the mixture. More precisely, the Boltzmann operators involve interactions between different species that will be weighted differently for each species, a mixing operates so that each species is balanced by its own weight.
\item[(3)] We emphasize here that the result still holds for any other global equilibrium $\mathbf{M}(c_{i,\infty},u_\infty,\theta_\infty)$ and also that the uniqueness has only been obtained in a perturbative regime, that is among the solutions written under the form $\mathbf{F} = \boldsymbol\mu +\mathbf{f}$ with $\norm{\mathbf{f}}_{L^\infty_{x,v}(\mathbf{w})}$ small \cite{BriDau}.
\item[(4)]In the case of identical masses and hard sphere collision kernels ($\gamma=b=1$) we recover the threshold $k_0$ which has recently been obtained in the mono-species case \cite{GMM}.
\end{enumerate}
\end{remark}
\bigskip

\section{Decomposition of the linear operator and toolbox}\label{sec:decomposition toolbox}

As noticed in \cite[Section 4]{GMM} for the mono-species Boltzmann linear operator and extended in \cite[Section 6]{BriDau} for the multi-species case, the linear operator $\mathbf{L}$ can be decompose into a regularizing operator and a ``small perturbation'' of the collision frequency $\boldsymbol \nu$, where $\boldsymbol \nu$ is a multiplicative operator defined by
\begin{equation}\label{nu}
\nu_i(v) = \sum\limits_{j=1}^N \nu_{ij}(v),
\end{equation}
with
$$\nu_{ij}(v) = C_{ij}^{\Phi}\int_{\R^3\times\mathbb{S}^{2}} b_{ij}\left(\mbox{cos}\:\theta\right)\abs{v-v_*}^\gamma \mu_j(v_*)\:d\sigma dv_*.$$
Each of the $\nu_{ij}$ could be seen as the collision frequency $\nu(v)$ of a single-species Boltzmann kernel with kernel $B_{ij}$. It is well-known (for instance \cite{Ce}\cite{CIP}\cite{Vi2}\cite{GMM}) that under our assumptions: $\nu_{ij}(v) \sim \langle v \rangle^\gamma$. We can be more explicit by making the change of variable $v_*\mapsto \sqrt{m_j}v_*$
$$\nu_{ij}(v) = \frac{C_{ij}^\Phi l_{b_{ij}}}{4\pi m_j^{\frac{1+\gamma}{2}}}\:\nu\left(\sqrt{m_j}v\right)$$
with $\nu$ being the frequency collision associated to the hard sphere kernel $B=1$. From \cite[Remark 4.1]{GMM} it follows the explicit equivalence
\begin{equation}\label{nui bound}
\frac{C_{ij}^\Phi l_{b_{ij}}}{m_j^{\frac{1+\gamma}{2}}}\max\br{m_j^{\gamma/2}\abs{v}^\gamma,\:\sqrt{\frac{2}{e\pi}}} \leq \nu_{ij}(v) \leq \frac{C_{ij}^\Phi l_{b_{ij}}}{m_j^{\frac{1+\gamma}{2}}}(m_j^{\gamma/2}\abs{v}^\gamma+2).
\end{equation}

We shall use the decomposition derived in \cite{BriDau}, that we recall now.

\bigskip
For $\delta$ in $(0,1)$, we consider $\Theta_\delta = \Theta_\delta(v,v_*,\sigma)$ in $C^\infty$ that is bounded by one everywhere, is exactly one on the set
$$\left\{\abs{v}\leq \delta^{-1}    \quad\mbox{and}\quad 2\delta\leq\abs{v-v_*}\leq \delta^{-1}    \quad\mbox{and}\quad \abs{\mbox{cos}\:\theta} \leq 1-2\delta \right\}$$
and whose support is included in
$$\left\{\abs{v}\leq 2\delta^{-1}    \quad\mbox{and}\quad \delta\leq\abs{v-v_*}\leq 2\delta^{-1}    \quad\mbox{and}\quad \abs{\mbox{cos}\:\theta} \leq 1-\delta \right\}.$$
We define the splitting
\begin{equation}\label{decomposition A B}
\mathbf{L} = -\boldsymbol\nu + \mathbf{B}^{(\delta)} +\mathbf{A}^{(\delta)},
\end{equation}
with, for all $i$ in $\br{1,\dots,N}$,
$$A^{(\delta)}_i (\mathbf{h}) (v) = \sum\limits_{j=1}^N C^\Phi_{ij}\int_{\R^3\times\mathbb{S}^2}\Theta_\delta\left[\mu^{'*}_jh'_i + \mu'_ih^{'*}_j - \mu_i h^*_j\right]b_{ij}\left(\mbox{cos}\:\theta\right)\abs{v-v_*}^\gamma\:d\sigma dv_*$$
and
\begin{equation}\label{Bidelta}
B^{(\delta)}_i (\mathbf{h}) (v) = \sum\limits_{j=1}^N C^\Phi_{ij}\int_{\R^3\times\mathbb{S}^2}\pa{1-\Theta_\delta}\left[\mu^{'*}_jh'_i + \mu'_ih^{'*}_j - \mu_i h^*_j\right]b_{ij}\left(\mbox{cos}\:\theta\right)\abs{v-v_*}^\gamma\:d\sigma dv_*.
\end{equation}

\bigskip
We have the following regularizing effect for $\mathbf{A^{(\delta)}}$.
\bigskip
\begin{lemma}\label{lem:control A}
Let $w_i=e^{\kappa_1(\sqrt{m_i}\abs{v})^{\kappa_2}}$ with $\kappa_1 >0$ and $\kappa_2$ in $(0,2)$ or $w_i=\langle \sqrt{m_i}v \rangle^k$ with $k> 5+\gamma$. Then for any $\beta>0$ and $\delta$ in $(0,1)$, there exists $C_A>0$ such that for all $\mathbf{f}$ in $L^\infty_{x,v}(\mathbf{w})$
$$\norm{\mathbf{A^{(\delta)}} \pa{\mathbf{f}}}_{L^\infty_{x,v}\pa{\langle v\rangle^\beta\boldsymbol\mu^{-1/2}}} \leq C_A\norm{\mathbf{f}}_{L^\infty_{x,v}(\mathbf{w})}.$$
The constant $C_A$ is constructive and only depends on $k$, $\beta$, $\delta$, $N$ and the collision kernels.
\end{lemma}
\bigskip

\begin{proof}[Proof of Lemma $\ref{lem:control A}$]
The operator $\mathbf{A^{(\delta)}}$ can be written as a kernel operator where the kernels $\mathbf{k^{(i),(\delta)}_A}$ are of compact support (see \cite[Lemma 6.2]{BriDau}):
$$\forall \: i\in \br{1,\dots,N},\quad A^{(\delta)}_i (\mathbf{f})(x,v) =\int_{\R^3} \langle \mathbf{k^{(i),(\delta)}_A}(v,v_*),\mathbf{f}(x,v_*)\rangle\:dv_*.$$
The desired estimate is therefore straightforward.
\end{proof}
\bigskip

\bigskip
The estimate on $\mathbf{B^{(\delta)}}$ is more delicate as it requires sharp estimates. Indeed, as one would like to control $\mathbf{B^{(\delta)}}$ by the collision frequency $\boldsymbol\nu$ one needs to derive an exact ratio between the latter two operators in our weighted spaces. Since our weights are radially symmetric, we need a technical lemma that gives an estimate of the collision operator when applied to radially symmetric functions. Such a symmetry brings more precise estimates on the operator. Note that this result is an extension of \cite[Lemma 4.6]{GMM} (proved in the case of hard spheres $\gamma=b=1$) not only to the multi-species framework but also to more general kernels.

\bigskip
\begin{lemma}\label{lem:radiallysym}
For $i$, $j$ in $\br{1,\dots,N}$, define
$$Q^+_{ij}(F,G) = \int_{\R^3\times\S^2} b_{ij}(\cos\theta)\abs{v-v_*}^\gamma F(v')G(v'_*)\:d\sigma dv_*.$$
Then for $F$ and $G$ radially symmetric functions in $L^1_v$ we have the following bound
$$\abs{Q^+_{ij}(F,G)(v)} \leq \int_0^{+\infty} \int_0^{+\infty} \mathbf{1}_{\br{m_i(r')^2 + m_j (r'_*)^2 \geq m_i r^2}}B(r,r',r'_*)\abs{F}(r')\abs{G}(r'_*)\:dr'dr'_*,$$
where we denote $r=\abs{v}$ and
$$B(r,r',r'_*) = 16 \pi^2b_{ij}^\infty\frac{(m_i+m_j)^2}{m_im_j^2} \frac{r'r'_*}{r\abs{r'-r'_*}^{1-\gamma}}\min\br{m_ir,m_jr_*,m_ir',m_jr'_*},$$
with the definition $r_* = \sqrt{m_im_j^{-1}(r')^2 + (r'_*)^2-m_im_j^{-1}r^2}$.
\end{lemma}
\bigskip

The proof of Lemma \ref{lem:radiallysym} is technical and closely follows the one of \cite[Lemma 4.6]{GMM}. We therefore leave it in Appendix \ref{appendix:proof lemma radially symmetric} for the sake of completeness and also because the case $m_i\neq m_j$ has to be handled carefully as we need to track down the constants precisely. We furthermore emphasize that dealing with more general kernels is delicate as, for instance, one has to decompose between large and small relative velocities to control terms like $\abs{v-v_*}^{1-\gamma}$, which disappear in the case of hard spheres $\gamma=1$ \cite{GMM}.
\par The previous lemma can be directly used to obtain a control on the full non-linear operator. This is a well-known result for the mono-species Boltzmann equation that the control on the non-linear operator implies a loss of weight $\nu(v)$ in the case of a polynomial weight and $\nu(v)^{1-c}$, with $c>0$, in the exponential case (see for instance \cite[Lemma 5.16]{GMM}). We prove here a similar result for the multi-species operator.

\bigskip
\begin{lemma}\label{lem:control Q}
Let $w_i=e^{\kappa_1(\sqrt{m_i}\abs{v})^{\kappa_2}}$ with $\kappa_1 >0$ and $\kappa_2$ in $(0,2)$ or $w_i=\langle \sqrt{m_i}v \rangle^k$ with $k> 5+\gamma$,  there exists $C_Q >0$ such that for  every $\mathbf{f}$
$$\sum\limits_{i=1}^N\norm{Q_i(\mathbf{f},\mathbf{g})}_{L^\infty_{x,v}\pa{w_i\nu_i^{-1+c(\mathbf{w})}}} \leq C_Q \norm{\mathbf{f}}_{L^\infty_{x,v}(\mathbf{w})}\norm{\mathbf{g}}_{L^\infty_{x,v}(\mathbf{w})},$$
The constant $C_Q$ is explicit and depends only on $w$, $N$, the masses $m_i$ and the kernels of the collision operator. The power $c(\mathbf{w})$ is zero when $\mathbf{w}$ is polynomial and can be taken equal to $\kappa_2'/\gamma$ for any $0\leq \kappa_2'<\kappa_2$ for $\mathbf{w}$ exponential.
\end{lemma}
\bigskip

\begin{proof}[Proof of Lemma \ref{lem:control Q}]
We remind the definition
$$Q_i(\mathbf{f},\mathbf{g}) = \sum\limits_{j=1}^N Q_{ij}(f_i,g_j).$$
It is therefore enough to prove the estimate for $Q_{ij}(f_i,g_j)$.
\par Firstly, since the collision kernels satisfy the Grad's cutoff assumption we can decompose $Q_{ij}$ into a loss and a gain term as
$$Q_{ij}(f_i,g_j) = C_{ij}^\Phi\cro{Q^+_{ij}(f_i,g_j) - Q^-_{ij}(f_i,g_j)}$$
with $Q^+_{ij}$ has been defined in Lemma \ref{lem:radiallysym} and
$$Q^-_{ij}(f_i,g_j)(v) = \pa{\int_{\R^3\times\S^2} \abs{v-v_*}^\gamma b_{ij}(\cos \theta)g_j^*\:d\sigma dv_*}f_i(v).$$
In this proof we will denote by $C$ any positive constant independent of $\mathbf{f}$, $\mathbf{g}$ and $v$. We shall bound the gain and the loss terms separately.

\bigskip
\textbf{Step 1: Estimate on the loss part.} This operator can be controlled rather easily. Indeed,
$$\abs{w_i\nu_i^{-1}(v)Q^-_{ij}(f_i,g_j)} \leq \norm{f_i}_{L^\infty_{x,v}(w_i)}\norm{g_j}_{L^\infty_{x,v}(w_j)} \nu_i^{-1}(v)\pa{l_{b_{ij}}\int_{\R^3}\frac{\abs{v-v_*}^\gamma}{w_j(v_*)}\:dv_*}.$$
To conclude the estimate on the loss term, we bound $\abs{v-v_*}^\gamma$ by $\langle v \rangle^\gamma \langle v_* \rangle^\gamma$ and we remind that $\langle v \rangle^\gamma \leq C\nu_i(v)$. The remaining integral in $v_*$ is finite due to the weights $w_j(v_*)$ considered here.

\bigskip
\textbf{Step 2: Estimate on the gain part for $\abs{v}\leq 1$.} Bounding crudely the gain term for $\abs{v}\leq 1$ we get
\begin{equation*}
\begin{split}
\abs{w_i\nu_i^{-1}(v)Q^+_{ij}(f_i,f_j)} \leq b^\infty_{ij}&w_i(1)\nu_i(v)^{-1}\pa{\int_{\R^3\times\S^2}\frac{\abs{v-v_*}^\gamma}{w_i(v')w_j(v'_*)}\:d\sigma dv_*}
\\&\times \norm{\mathbf{f}}_{L^\infty_{x,v}(\mathbf{w})}\norm{\mathbf{g}}_{L^\infty_{x,v}(\mathbf{w})}.
\end{split}
\end{equation*}
First we bound as before $\abs{v-v_*}^\gamma$ by $C\nu_i(v)\langle v_* \rangle^\gamma$. Then we use the energy conservation of elastic collisions to see that $m_i\abs{v'}^2 + m_j\abs{v'_*}^2 \geq m_j \abs{v_*}^2$ and therefore that $\abs{v'}\geq \sqrt{m_j/(2m_i)}\abs{v_*}$ or $\abs{v'_*} \geq \sqrt{1/2}\abs{v_*}$. Since $v \mapsto 1/w_{i/j}(\abs{v})$ is decreasing and also that $w(\abs{v}) \geq 1$ it follows that the following always holds
\begin{equation}\label{wiwj}
\frac{\abs{v-v_*}^\gamma}{w_i(v')w_j(v'_*)} \leq C \nu_i(v)\frac{\langle v_* \rangle^\gamma}{w\pa{\min\br{\sqrt{\frac{m_i}{2}},\sqrt{\frac{m_j}{2}}}v_*}}.
\end{equation}
We infer
\begin{equation}\label{control Q v leq 1}
\forall \abs{v}\leq 1, \: \abs{w_i\nu_i^{-1}(v)Q^+_{ij}(f_i,g_j)} \leq C\max\limits_{1\leq i \leq N}\pa{\int_{\R^3}\frac{\langle v_*\rangle^\gamma}{w(\sqrt{\frac{m_i}{2}}v_*)}\:dv_*}\norm{\mathbf{f}}_{L^\infty_{x,v}(\mathbf{w})}\norm{\mathbf{g}}_{L^\infty_{x,v}(\mathbf{w})}.
\end{equation}

\bigskip
\textbf{Step 3.1: Estimate on the gain part for $\abs{v}\geq 1$  when $w_i$ is poynomial.} Bounding $f_i$ and $f_j$ by their $L^\infty_{x,v}(\mathbf{w})$-norm and since $\mathbf{w}(v)=\mathbf{w}(\abs{v})$ we use Lemma \ref{lem:radiallysym} and obtain
\begin{equation}\label{Qij+ v geq 1 first bound}
\abs{w_i\nu_i^{-1}(v)Q^+_{ij}(f_i,f_j)} \leq C w_i(v)\nu_i(v)^{-1} I(\abs{v})\norm{\mathbf{f}}_{L^\infty_{x,v}(\mathbf{w})}\norm{\mathbf{g}}_{L^\infty_{x,v}(\mathbf{w})}
\end{equation}
with
\begin{equation}\label{I}
I(r) = \int_0^{+\infty} \int_0^{+\infty} \mathbf{1}_{A_{ij}(r)}\frac{r'r'_*}{r\abs{r'-r'_*}^{1-\gamma}w_i(r')w_j(r'_*)}\min\br{m_ir,m_jr_*,m_ir',m_jr'_*}dr'dr'_*
\end{equation}
and $$A_{ij}(r) = \br{mi(r')^2+m_j(r'_*)^2\geq m_i r^2} \subset \br{r'\geq \frac{1}{\sqrt{2}}r} \cup \br{r'_*\geq \sqrt{\frac{m_i}{2m_j}r}}.$$

\par We decompose $I(r)$ into two integrals. The first one when $r'\geq r/\sqrt{2}$ and $r'_*\geq 0$, on which we bound $\min\{m_ir,m_jr_*,m_ir',m_jr'_*\}$ by $m_jr'_*$ and the second one when $r'_* \geq \sqrt{m_i/2m_j}r$ and $r'\geq 0$, on which we bound the minimum by $m_ir'$. This yields
\begin{equation*}
\begin{split}
I(r)\leq& \frac{C}{r} \int_{\frac{1}{\sqrt{2}}r}^{+\infty}\frac{r'}{w_i(r')}\pa{\int_0^{+\infty} \frac{(r'_*)^2}{\abs{r'-r'_*}^{1-\gamma}w_j(r'_*)}dr'_*}dr'
\\&+\frac{C}{r} \int_{\frac{\sqrt{m_i}}{\sqrt{2m_j}}r}^{+\infty}\frac{r'_*}{w_j(r'_*)}\pa{\int_0^{+\infty} \frac{(r')^2}{\abs{r'-r'_*}^{1-\gamma}w_i(r')}dr'}dr'_*
\\\leq& \frac{C}{r} \int_{ar}^{+\infty}\frac{r'}{(1+(r')^2)^{k/2}}\pa{\int_0^{+\infty} \frac{(r'_*)^2}{\abs{r'-r'_*}^{1-\gamma}(1+(r'_*)^2)^{k/2}}dr'_*}dr',
\end{split}
\end{equation*}
where we defined $a = \min\br{\sqrt{1/2},\sqrt{m_i/2m_j}}$ and we used that there exists $C>0$ such that for all $i$, $w_i(r)^{-1} \leq C/(1+r^2)^{k/2}$. Finally, we define $R = ar/2$ and we decompose the integral in $r'_*$ into the part where $\abs{r'-r'_*}\leq R$ and the part where $\abs{r'-r'_*}\geq R$. This yields
\begin{equation}\label{important bound I}
\begin{split}
I(r) \leq& \frac{C}{r}\int_{ar}^{+\infty}\frac{r'}{(1+(r')^2)^{k/2}}\pa{\int_{r'-R}^{r'+R} \frac{(r'_*)^2}{\abs{r'-r'_*}^{1-\gamma}(1+(r'_*)^2)^{k/2}}dr'_*}dr'
\\&+ \frac{C}{rR^{1-\gamma}}\pa{\int_{ar}^{+\infty}\frac{r'}{(1+(r')^2)^{k/2}}dr'}\pa{\int_0^{+\infty} \frac{(r'_*)^2}{(1+(r'_*)^2)^{k/2}}dr'_*}.
\end{split}
\end{equation}

\bigskip
Since for all $p>0$, $1/(1+x^2)^{p}$ is decreasing for $x\geq 0$ and since $r'_* \geq r'-R\geq ar-R\geq 0$, in the first term on the right-hand side of $\eqref{important bound I}$ we bound $(r'_*)^2/(1+(r'_*)^2)^{k/2}$ first by $C/(1+(r'_*)^2)^{(k-2)/2}$ and then by $1/(1+(ar-R)^2)^{(k-2)/2}$. We compute directly the second term on the right-hand side of $\eqref{important bound I}$.
\begin{eqnarray*}
I(r) &\leq& \frac{C}{r(1+(ar-R)^2)^{(k-2)/2}}\pa{\int_{ar}^{+\infty}\frac{r'}{(1+(r')^2)^{k/2}}dr'}\pa{\int_{-R}^R \frac{dr'_*}{\abs{r'_*}^{1-\gamma}}} 
\\&\quad&+ \frac{C}{rR^{1-\gamma}}\pa{\int_{ar}^{+\infty}\frac{r'}{(1+(r')^2)^{k/2}}dr'}\pa{\int_{0}^{+\infty} \frac{(r'_*)^2}{(1+(r'_*)^2)^{k/2}}dr'_*} \nonumber
\\&\leq& C R^\gamma\cro{\frac{1}{r(1+(ar-R)^2)^{(k-2)/2}(1+(ar)^2)^{(k-2)/2}} + \frac{1}{rR(1+(ar)^2)^{(k-2)/2}}}. \nonumber
\end{eqnarray*}
Now we use $R=ar/2$, the fact that $r^\gamma \leq C \nu_i(v)$ and the  inequality $r^2(1+(ar)^2)^p \geq C (1+r^2)^{p+1}$ for $a$ and $p$ nonnegative and $r\geq 1$. We deduce
$$I(\abs{v}) \leq C \frac{\nu_i(v)}{w_i(v)}\cro{\frac{1}{(1+r^2)^{(k-2)/2}}+1},$$
and therefore, with $\eqref{Qij+ v geq 1 first bound}$ we obtain
\begin{equation}\label{Qij+ polynomial}
\forall \abs{v}\geq 1,\quad \abs{w_i\nu_i^{-1}(v)Q^+_{ij}(f_i,g_j)} \leq C\norm{\mathbf{f}}_{L^{\infty}_{x,v}(\mathbf{w})}\norm{\mathbf{g}}_{L^\infty_{x,v}(\mathbf{w})}.
\end{equation}

\bigskip
\textbf{Step 3.2: Estimate on the gain part for $\abs{v}\geq 1$  when $w_i$ is exponential.}
Let $0 < \kappa_2' < \kappa_2$. We start with $\eqref{Qij+ v geq 1 first bound}$ that is
\begin{equation}\label{controlQ+ expo start}
\abs{w_i\nu_i^{-1+\frac{\kappa_2'}{\gamma}}(v)Q^+_{ij}(f_i,g_j)} \leq C w_i(v)\nu_i(v)^{-1+\frac{\kappa_2'}{\gamma}} I(\abs{v})\norm{\mathbf{f}}_{L^\infty_{x,v}(\mathbf{w})}\norm{\mathbf{g}}_{L^\infty_{x,v}(\mathbf{w})}
\end{equation}
where $I(r)$ is defined in $\eqref{I}$. We make the following change of variable $(r',r'_*) \mapsto (r'/\sqrt{m_i},r'_*/\sqrt{m_j})$ which yields
\begin{equation*}
\begin{split}
I(r) \leq  C \int_0^{+\infty} \int_0^{+\infty}& \mathbf{1}_{\br{(r')^2+(r'_*)^2 \geq m_i r^2}}\frac{r'r'_*\min\br{r',r'_*}}{r\abs{m_i^{-1/2}r'-m_j^{-1/2}r'_*}^{1-\gamma}}e^{-\kappa_1((r')^{\kappa_2}+(r'_*)^{\kappa_2})} dr'dr'_*.
\end{split}
\end{equation*}
Note that we bounded the minimum of four terms by the minimum of two terms, also we bounded the masses by their maximum.
\par We decompose the integral on the right-hand side into $\br{r' \geq r'_*}$, on which we bound the minimum by $r'_*$, and $\br{r'_* \geq r'}$, where the minimum is bounded from above by $r'$. The two integrals are equal after the relabelling $(r',r'_*)$ into $(r'_*,r')$ and thus
$$I(r) \leq C \int_0^{+\infty} \int_{r'}^{+\infty}  \mathbf{1}_{\br{(r')^2+(r'_*)^2 \geq m_i r^2}}\frac{(r')^2 r'_*}{r\abs{m_i^{-1/2}r'-m_j^{-1/2}r'_*}^{1-\gamma}}e^{-\kappa_1((r')^{\kappa_2}+(r'_*)^{\kappa_2})} dr'_*dr'.$$
From \cite[Proof of Lemma 4.10]{GMM}, if we denote $\rho = \sqrt{(r')^2+(r'_*)^2}$ we note that $r'_* \geq r'$ implies $r' \leq \rho /\sqrt{2}$ and hence the following upper bound 
$$e^{-\kappa_1((r')^{\kappa_2}+(r'_*)^{\kappa_2})} \leq e^{-\kappa_1\rho^{\kappa_2}}e^{-\kappa_1 \eta (r')^{\kappa_2}},$$
where $\eta$ only depends on $\kappa_2$. We make the change of variable $(r',r'_*)\mapsto (r',\rho)$ (recall $r' \leq \rho /\sqrt{2}$) 
$$I(r) \leq \frac{C}{r} \int_{\sqrt{m_i}r}^{+\infty} \rho e^{-\kappa_1 \rho^{\kappa_2}} \int_0^{\frac{\rho}{\sqrt{2}}} e^{-\kappa_1\eta(r')^{\kappa_2}} \frac{(r')^2}{\abs{m_i^{-1/2}r' -m_j^{-1/2}\sqrt{\rho^2-(r')^2}}^{1-\gamma}}\:dr'd\rho,$$
which we can rewrite thanks to the change of variable $r' \mapsto a\rho$, with $a$ belonging to $[0,1/\sqrt{2}]$:
$$I(r) \leq \frac{C}{r} \int_{\sqrt{m_i}r}^{+\infty} \rho^{\kappa_2-\kappa_2'+\gamma} e^{-\kappa_1 \rho^{\kappa_2}} \int_0^{\frac{1}{\sqrt{2}}} e^{-\kappa_1\eta(a\rho)^{\kappa_2}} \frac{\rho^{3-(\kappa_2-\kappa_2')} a^2}{\abs{m_i^{-1/2}a -m_j^{-1/2}\sqrt{1-a^2}}^{1-\gamma}}\:dad\rho.$$
\par For any $0<a<1/\sqrt{2}$, the following holds
$$\rho^{3-(\kappa_2-\kappa_2')} e^{-\kappa_1\eta(a\rho)^{\kappa_2}}  \leq \frac{C}{a^{3-(\kappa_2-\kappa_{2'})}}e^{-\frac{\kappa_1\eta}{2} (\sqrt{m_i}a)^{\kappa_2}}.$$ 
The integrand in $a$ variable is therefore uniformly bounded in $\rho$ by an integrable function (because $0\leq 1-\gamma <1$ and $-1 < 1-(\kappa_2-\kappa_2') < 1$). Hence, by integrating by part,
\begin{eqnarray*}
I(r) &\leq& \frac{C}{r}\int_{\sqrt{m_i}r}^{+\infty} \rho^{\kappa_2-\kappa_2'+\gamma} e^{-\kappa_1 \rho^{\kappa_2}}\:d\rho = \frac{C}{r}\int_{\sqrt{m_i}r}^{+\infty} \rho^{1+\gamma-\kappa'_2}\rho^{\kappa_2-1} e^{-\kappa_1 \rho^{\kappa_2}}\:d\rho
\\ &\leq& Cr^{\gamma-\kappa'_2}e^{-\kappa_1 (\sqrt{m_i}r)^{\kappa_2}} + C\int_{\sqrt{m_i}r}^{+\infty}\rho^{\gamma-\kappa_2'}e^{-\kappa_1\rho^{\kappa_2}}\:d\rho.
\end{eqnarray*}
To conclude we integrate by part inductively until the power of $\rho$ is negative inside the integral. For $0 \leq b \leq \gamma-\kappa_2'$ we have $r^{b} \leq r^{\gamma-\kappa_2'}$ since $r\geq 1$. Recalling that $\nu_i(v) \sim (1+\abs{v}^2)^{\gamma/2}$ (see $\eqref{nui bound}$) it follows that for $\abs{v} \geq 1$,
$$I(\abs{v}) \leq C w_i(v)^{-1}\abs{v}^{\gamma-\kappa_2'} \leq C w_i(v)^{-1}\nu_i(v)^{1-\frac{\kappa_2'}{\gamma}}.$$
Plugging the above into $\eqref{controlQ+ expo start}$ terminates the proof.
\end{proof}
\bigskip

We conclude this section with an explicit control for the operator $\mathbf{B^{(\delta)}}$.

\bigskip
\begin{lemma}\label{lem:control B}
Let $w_i=e^{\kappa_1(\sqrt{m_i}\abs{v})^{\kappa_2}}$ with $\kappa_1 >0$ and $\kappa_2$ in $(0,2)$ or $w_i=\langle \sqrt{m_i}v \rangle^k$ with $k> k_0$. Take $\delta$ be in $(0,1)$. Then for all $\mathbf{f}$ in $L^\infty_{x,v}\pa{w\boldsymbol\nu}$ and all $i$ in $\br{1,\dots,N}$,
$$\sum\limits_{i=1}^N\norm{B^{(\delta)}_i(\mathbf{f})}_{L^\infty_{x,v}\pa{w_i\nu_i^{-1}}} \leq C_B(\mathbf{w},\delta)\norm{\mathbf{f}}_{L^\infty_{x,v}(\mathbf{w})}.$$
Moreover we have the following formula
$$C_B(\mathbf{w},\delta) = C_B(\mathbf{w}) + \eps_w(\delta)$$
where $\eps_w(\delta)$ is an explicit function depending on $w$ that tends to $0$ as $\delta$ tends to $0$ and
\begin{itemize}
\item[(i)] in the case $w_i=e^{\kappa_1\abs{\sqrt{m_i}v}^{\kappa_2}}$: $C_B(\mathbf{w}) = 0$; 
\item[(ii)] in the case $w_i=\langle \sqrt{m_i}v \rangle^k$:
$$C_B(\mathbf{w}) = \frac{4\pi}{k-1-\gamma}\max\limits_{1\leq i \leq N}\cro{\pa{\sum\limits_{j=1}^NC^\Phi_{ij}b_{ij}^\infty\frac{(m_i+m_j)^2}{m_i^{2-\frac{\gamma}{2}}m_j^{\frac{5+\gamma}{2}}}}\pa{\sum\limits_{1\leq k \leq N}\frac{\sqrt{m_k}}{C^\Phi_{ik}l_{b_{ik}}}}}.$$
\end{itemize}
\end{lemma}
\bigskip

\begin{remark}\label{rem:CB<1}
We emphasize here that, by definition of $k_0$ given by $\eqref{k0}$, for every choice of weight $w$ considered in this work one has $C_B(\mathbf{w})<1$. We can thus fix $\delta_0$ small enough such that for all $\delta \leq \delta_0$, $C_B = C_B(\mathbf{w},\delta) <1$. In the rest of this article we will assume that we chose $\delta \leq \delta_0$ and for convenience we will drop the exponent and use the following notations: $\mathbf{B}=\mathbf{B^{(\delta)}}$, $\mathbf{A}=\mathbf{A^{(\delta)}}$ and finally $C_B = C_B(\mathbf{w},\delta)$.
\end{remark}
\bigskip

\begin{proof}[Proof of Lemma \ref{lem:control B}]
Following the idea of \cite{GMM} for the mono-species case in the hard spheres model we split the operator $B_i^{(\delta)}$ (recall $\eqref{Bidelta}$) into three pieces.
\begin{equation}\label{Bistart}
\begin{split}
\abs{B_i^{(\delta)}(\mathbf{f})} \leq& \sum\limits_{j=1}^N C_{ij}^{\Phi}b_{ij}^\infty \int_{\R^3\times\S^2} \mathbf{1}_{\abs{v}\geq R} \abs{v-v_*}^\gamma \pa{\mu_j^{'*}\abs{f_i'}+\mu_i'\abs{f_j^{'*}}}\:d\sigma dv_*
\\&+ \sum\limits_{j=1}^N C_{ij}^{\Phi}b_{ij}^\infty \int_{\R^3\times\S^2}\mathbf{1}_{\abs{v}\leq R}\pa{1-\Theta_\delta}\abs{v-v_*}^\gamma \pa{\mu_j^{'*}\abs{f_i'}+\mu_i'\abs{f_j^{'*}}}\:d\sigma dv_*
\\&+ \sum\limits_{j=1}^N C_{ij}^{\Phi}b_{ij}^\infty \int_{\R^3\times\S^2}\pa{1-\Theta_\delta}\abs{v-v_*}^\gamma \mu_i \abs{f_j^*}\:d\sigma dv_*.
\end{split}
\end{equation}
The last two terms are easily handled. Indeed, we use $\abs{v-v_*}^\gamma \leq C \nu_i(v)\langle v_* \rangle^\gamma$ for the last one:
\begin{equation}\label{last 2 terms CV 0}
\begin{split}
w_i(v)&\nu_i^{-1}(v)\int_{\R^3\times\S^2}\pa{1-\Theta_\delta}\abs{v-v_*}^\gamma \mu_i \abs{f_j^*}\:d\sigma dv_* 
\\&\leq C w_i(v)\mu_i(v)\mathbf{1}_{\abs{v}\geq R}\pa{\int_{\R^3\times\S^2}\frac{\langle v_* \rangle^\gamma}{w_j(v_*)}\:dv_*d\sigma} \norm{\mathbf{f}}_{L^\infty_{x,v}(\mathbf{w})}
\\&\quad + C w_i(R) \pa{\int_{\R^3\times\S^2}\mathbf{1}_{\abs{v}\leq R}\pa{1-\Theta_\delta}\frac{\langle v_*\rangle^\gamma}{w_j(v_*)}\:dv_*d\sigma} \norm{\mathbf{f}}_{L^\infty_{x,v}(\mathbf{w})}.
\end{split}
\end{equation}
Then notice that for all $i,\:j$, 
$$\abs{v-v_*}^\gamma\frac{\mu_{i}(v')}{w_{j}(v'_*)} \leq \frac{C}{w_i(v')w_j(v'_*)} \leq C\nu_i(v)\frac{\langle v_* \rangle^\gamma}{w(\min\br{\sqrt{\frac{m_i}{2}}}v)}$$
where we used $\eqref{wiwj}$. As the same holds when exchanging $i$ with $j$ and $v'$ with $v'_*$ we can handle the second term in $\eqref{Bistart}$ by taking the $L^\infty_{x,v}(\mathbf{w})$-norm of $\mathbf{f}$ out:
\begin{equation}\label{last 2 terms CV 1}
\begin{split}
w_i(v)&\nu_i^{-1}(v)\int_{\R^3\times\S^2}\mathbf{1}_{\abs{v}\leq R}\pa{1-\Theta_\delta}\abs{v-v_*}^\gamma \pa{\mu_j^{'*}\abs{f_i'}+\mu_i'\abs{f_j^{'*}}}\:d\sigma dv_* 
\\&\leq C w_i\pa{R}\pa{\int_{\R^3\times\S^2}\mathbf{1}_{\abs{v}\leq R}\pa{1-\Theta_\delta}\frac{\langle v_*\rangle^\gamma}{w(\min\br{\sqrt{m_i/2}}v_*)}\:dv_*d\sigma} \norm{\mathbf{f}}_{L^\infty_{x,v}(\mathbf{w})}.
\end{split}
\end{equation}
One has for $\delta^{-1} \geq 2R$ 
\begin{equation*}
\begin{split}
\mathbf{1}_{\abs{v}\leq R}\pa{1-\Theta_\delta} &\leq \mathbf{1}_{\abs{v}\leq R}\mathbf{1}_{\abs{v-v_*}\geq \delta^{-1}} + \mathbf{1}_{\abs{v-v_*}\leq 2\delta} + \mathbf{1}_{\abs{\cos \theta} \geq 1-2\delta} 
\\&\leq \mathbf{1}_{\abs{v_*}\geq \delta^{-1}/2} + \mathbf{1}_{\abs{v-v_*}\leq 2\delta} + \mathbf{1}_{\abs{\cos \theta} \geq 1-2\delta}.
\end{split}
\end{equation*}
Firstly, $w_i(v)\mu_i(v) \mathbf{1}_{\abs{v}\geq R}$ tends to zero as $R$ goes to infinity. Secondly, the function $\phi: v_* \mapsto\langle v_*\rangle^\gamma/ w(v_*)$ is integrable on $\R^3$ thus the integral of $\mathbf{1}_{\abs{v}\leq R}\pa{1-\Theta_\delta}\phi(v_*)$ tends to zero for fixed $R$ as $\delta$ goes to zero. Therefore, $\eqref{last 2 terms CV 0}$ and $\eqref{last 2 terms CV 1}$ both tend to zero. Hence $\eqref{Bistart}$ becomes
\begin{equation}\label{Bidelta general}
\begin{split}
w_i\nu_i^{-1}\abs{B_i^{(\delta)}(\mathbf{f})} \leq& \pa{\sum\limits_{j=1}^N\norm{\pa{Q^+_{ij}(\mu_i,w_j^{-1})+Q^+_{ij}(w_i^{-1},\mu_j)}\mathbf{1}_{\abs{v}\geq R}}_{L^\infty_{x,v}(w_i\nu_i^{-1})}}\norm{\mathbf{f}}_{L^\infty_{x,v}(\mathbf{w})}
\\&+ \pa{\eps_w(R) + w_i(R)\eps_w(\delta)}\norm{\mathbf{f}}_{L^\infty_{x,v}(\mathbf{w})} .
\end{split}
\end{equation}

\bigskip
The conclusion in the case of $w_i(v)$ being exponential is direct from Lemma \ref{lem:control Q} (more precisely, the control of $Q^+_i$ which is the same as the whole $Q_i$) since  the gain of weight $\nu_i^{c(\mathbf{w})}$ implies that the sum in $\eqref{Bidelta general}$ tends to zero when $R$ tends to infinity.
\par The case where $w_i(v)=\langle \sqrt{m_i}v\rangle^k$ is a bit more computational since the control of the bilinear term $Q_{ij}^+(\mu_i,w_j^{-1})$ does not tend to $0$ as $\abs{v}$ goes to infinity. It requires explicit estimates from Step 3.1 in the proof of Lemma \ref{lem:control Q}. Applying Lemma \ref{lem:radiallysym} we have
\begin{equation*}
\begin{split}
&\abs{Q^+_{ij}(\mu_i,w_j^{-1}) (v)}
\\&\quad\quad= C_0\int_0^{+\infty} \int_0^{+\infty} \mathbf{1}_{A_{ij}(r)}\frac{r'r'_*\mu_i(r')}{r\abs{r'-r'_*}^{1-\gamma}w_j(r'_*)}\min\br{m_ir,m_jr_*,m_ir',m_jr'_*}dr'dr'_*
\end{split}
\end{equation*}
where we recall that 
\begin{equation}\label{def C0}
C_0 = 16\pi^2C^\Phi_{ij} b^\infty_{ij}(m_i+m_j)^2/(m_im_j^2)
\end{equation}
and 
$$A_{ij}(r) = \br{mi(r')^2+m_j(r'_*)^2\geq m_i r^2} \subset \br{r'\geq \sqrt{\eps}r} \cup \br{r'_*\geq \sqrt{\frac{m_i(1-\eps)}{m_j}}\:r},$$
which holds for any $\eps$ in $(0,1)$. As for $\eqref{I}$ we decompose into two integrals on each of the subsets above. On the first one we bound the minimum by $m_jr'_*$ while we bound it by $m_ir'$ on the second set. This yields
\begin{equation*}
\begin{split}
\abs{Q^+_{ij}(\mu_i,w_j^{-1})}&\leq \frac{C_0}{(2\pi)^{3/2}r} \int_{\sqrt{\eps}r}^{+\infty}dr'r'e^{-m_i\frac{(r')^2}{2}}\pa{\int_0^{+\infty} \frac{m_j(r'_*)^2dr'_*}{\abs{r'-r'_*}^{1-\gamma}(1+m_j(r'_*)^2)^{k/2}}}
\\&\:+m_i\frac{C_0}{(2\pi)^{3/2}r} \int_{\sqrt{\frac{m_i(1-\eps)}{m_j}}r}^{+\infty}\frac{r'_*dr'_* }{(1+m_j(r'_*)^2)^{k/2}}\pa{\int_0^{+\infty} \frac{(r')^2e^{-m_i\frac{(r')^2}{2}}}{\abs{r'-r'_*}^{1-\gamma}}dr'}
\\&= I_1^{(\eps)} + I_2^{(\eps)}.
\end{split}
\end{equation*}
$I_1^{(\eps)}$ is easily dealt with using the same techniques as for $\eqref{important bound I}$ with $R = \sqrt{\eps}r/2$. We infer that there exists $C>0$ independent of $r$ and $\eps$ such that
\begin{eqnarray*}
I_1^{(\eps)} &\leq& C\pa{\frac{1}{(1+\eps r^2)^{(k-2)/2}} + \frac{1}{\eps^{(1-\gamma)/2}r^{(1+\gamma)/2}}}\int_{\sqrt{\eps}r}^{+\infty}r'e^{-m_i\frac{(r')^2}{2}}\:dr'
\\&\leq& C\pa{\frac{1}{(1+\eps r^2)^{(k-2)/2}} + \frac{1}{\eps^{(1-\gamma)/2}r^{(1+\gamma)/2}}}e^{-m_i\frac{\eps r^2}{2}}
\end{eqnarray*}
which implies, for any $\eps >0$,
\begin{equation}\label{bound I1}
\lim\limits_{R\to 0}\norm{I_1^{(\eps)}(v)\mathbf{1}_{\abs{v}\geq R}}_{L^\infty_v(w_i\nu_i^{-1})} = 0
\end{equation}
\par We denote by $C$ any positive constant independent of $r$. We decompose $I^{(\eps)}_2$ into an integral over $\br{\abs{r'-r'_*}\leq \eta\sqrt{\frac{m_i(1-\eps)}{m_j}}r}$ and its complementary, where $0<\eta<1$. Note that from $\abs{r'-r'_*}\leq \eta\sqrt{\frac{m_i(1-\eps)}{m_j}}r$ one can deduce $r' \geq (1-\eta)\sqrt{\frac{m_i(1-\eps)}{m_j}}r$ and thus
$$(r')^2e^{-m_i\frac{(r')^2}{2}} \leq C e^{-m_i\frac{(r')^2}{4}} \leq e^{-\frac{m_i^2(1-\eps)}{m_j}(1-\eta)^2\frac{r^2}{4}}.$$
We get
\begin{equation*}
\begin{split}
I_2^{(\eps)} \leq& \frac{C}{r}\pa{\int_0^{+\infty}\frac{r'_*}{(1+m_j(r'_*)^2)^{k/2}}dr'_*}\pa{\int_{0}^{\eta\sqrt{\frac{m_i(1-\eps)}{m_j}}r}\frac{dr'}{(r')^{1-\gamma}}}e^{-\frac{m_i^2(1-\eps)}{m_j}(1-\eta)^2\frac{r^2}{4}}
\\&+ \frac{m_i^{\frac{1+\gamma}{2}}m_j^{\frac{1-\gamma}{2}}C_0}{(2\pi)^{3/2}(1-\eps)^{\frac{1-\gamma}{2}}\eta^{1-\gamma}r^{2-\gamma}}\pa{\int_{\sqrt{\frac{m_i(1-\eps)}{m_j}}r}^{+\infty}\frac{r'_*}{(1+m_j(r'_*)^2)^{\frac{k}{2}}}dr'_*}
\\&\quad\quad\times \pa{\int_0^{+\infty}(r')^2e^{-m_i\frac{(r')^2}{2}}dr'}
\\\leq& C r^{\gamma-1} e^{-\frac{m_i^2(1-\eps)}{m_j}(1-\eta)^2\frac{r^2}{4}} + \frac{C_0 (1+m_i(1-\eps)r^2)^{\frac{(1-\gamma)}{2}}}{4\pi m_i^{1-\frac{\gamma}{2}}m_j^{\frac{1+\gamma}{2}}(1-\eps)^{\frac{1-\gamma}{2}}\eta^{1-\gamma}r^{2-\gamma}}
\\&\quad\quad\quad\quad\quad\quad\quad\quad\quad\quad\times \pa{\int_{\sqrt{m_i(1-\eps)}r}^{+\infty}\frac{u}{(1+u^2)^{\frac{k+1-\gamma}{2}}}du}
\\=& C e^{-\frac{m_i^2(1-\eps)}{m_j}(1-\eta)^2\frac{r^2}{4}} + \frac{4\pi}{k-1-\gamma}\frac{(m_i+m_j)^2}{m_i^{2-\frac{\gamma}{2}}m_j^{\frac{5+\gamma}{2}}}\frac{C^\Phi_{ij}b^\infty_{ij}\eta^{\gamma-1}}{(1-\eps)^{\frac{1-\gamma}{2}}r^{2-\gamma}(1+m_i(1-\eps)r^2)^{\frac{k}{2}-1}},
\end{split}
\end{equation*}
where we used that $r \geq R \geq 1$ and the definition $\eqref{def C0}$ of $C_0$. As $\nu_i = \sum_k \nu_{ik}$, we use the lower bound $\eqref{nui bound}$ and obtain for $\abs{v}\geq R \geq 1$
\begin{equation*}
\begin{split}
I_2^{(\eps)} \leq&  w_i^{-1}\nu_i \frac{4\pi}{k-1-\gamma}\frac{(m_i+m_j)^2}{m_i^{2-\frac{\gamma}{2}}m_j^{\frac{5+\gamma}{2}}}\frac{\eta^{\gamma-1}\pa{\sum\limits_{1\leq k \leq N}\frac{\sqrt{m_k}}{C^\Phi_{ik}l_{b_{ik}}}}C^\Phi_{ij}b_{ij}^\infty}{(1-\eps)^{\frac{1-\gamma}{2}}}\frac{(1+m_ir^2)^{\frac{k}{2}}}{r^{2}(1+m_i(1-\eps)r^2)^{\frac{k}{2}-1}}
\\&+w_i^{-1}\nu_i \cro{ C (1+r^2)^{\frac{k}{2}-\gamma} e^{-\frac{m_i(1-\eps)}{m_j}\frac{r^2}{4}}}.
\end{split}
\end{equation*}
This indicates that for all $\eps$ in $(0,1)$,
\begin{equation}\label{bound I2}
\lim\limits_{R\to +\infty}\norm{I_2^{(\eps)}(v)\mathbf{1}_{\abs{v}\geq R}}_{L^\infty_v(w_i\nu_i^{-1})} = \frac{4\pi}{k-1-\gamma}\frac{(m_i+m_j)^2}{m_i^{2-\frac{\gamma}{2}}m_j^{\frac{5+\gamma}{2}}}C^\Phi_{ij}b_{ij}^\infty\pa{\sum\limits_{1\leq k \leq N}\frac{\sqrt{m_k}}{C^\Phi_{ik}l_{b_{ik}}}}\frac{\eta^{\gamma-1}}{(1-\eps)^{\frac{1-\gamma}{2}}}.
\end{equation}
From $\eqref{bound I1}$ and $\eqref{bound I2}$, which holds for any $\eps$ and $\eta$ in $(0,1)$, we can bound $\eqref{Bidelta general}$ as the stated in by the lemma for $w_i$ polynomial. This concludes the proof.
\end{proof}
\bigskip

\section{Exponential decay of solutions with bounded initial data}\label{sec:exponential decay}

Let $w_i=e^{\kappa_1(\sqrt{m_i}\abs{v})^{\kappa_2}}$ with $\kappa_1 >0$ and $\kappa_2$ in $(0,2)$ or $w=\langle \sqrt{m_i}v \rangle^k$ with $k> k_0$, defined by $\eqref{k0}$. Let $\beta >3/2$.
\par As described in the introduction, we look for a solution $\mathbf{f}(t,x,v)$ in $L^\infty_{x,v}\pa{\mathbf{w}}$ of the Boltzmann equation
\begin{equation}\label{perturbedmultiBE expodecay}
\left\{\begin{array}{l}\disp{\partial_t \mathbf{f}+ v\cdot\nabla_x \mathbf{f} = \mathbf{L}\pa{\mathbf{f}} +\mathbf{Q}\pa{\mathbf{f}}}\vspace{2mm}\\\vspace{2mm} \disp{\mathbf{f}(0,x,v) = \mathbf{f_0}(x,v)}\end{array}\right.
\end{equation}
in the form of $\mathbf{f}=\mathbf{f_1}+\mathbf{f_2}$. We look for $\mathbf{f_1}$ in $L^\infty_{x,v}\pa{\mathbf{w}}$ and $f_2$ in $L^\infty_{x,v}\pa{\langle v \rangle^\beta\boldsymbol\mu^{-1/2}}\subset L^\infty_{x,v}\pa{\mathbf{w}}$ (due to the weights $\mathbf{w}$ considered here) and $(\mathbf{f_1},\mathbf{f_2})$ satisfying the following system of equations
\begin{eqnarray}
&&\partial_t \mathbf{f_1} +v\cdot\nabla_x\mathbf{f_1} = -\boldsymbol\nu(v)\mathbf{f_1} + \mathbf{B}\pa{\mathbf{f_1}}+ \mathbf{Q}(\mathbf{f_1}+\mathbf{f_2}) \quad\mbox{and}\quad \mathbf{f_1}(0,x,v)=\mathbf{f_0}(x,v),\label{f1}
\\&&\partial_t \mathbf{f_2}+v\cdot\nabla_x\mathbf{f_2} = \mathbf{L}(\mathbf{f_2}) + \mathbf{A}(\mathbf{f_1}) \quad\mbox{and}\quad \mathbf{f_2}(0,x,v)=0\label{f2}.
\end{eqnarray}
The operators $\mathbf{A}$ and $\mathbf{B}$ have been defined in $\eqref{decomposition A B}$ with $\delta$ as described in Remark \ref{rem:CB<1}.

\bigskip
\cite[Theorem 2.2]{BriDau} shows that for some $n$ in $(2,3)$ there exists $\eta >0$ such that if $\norm{\mathbf{f_0}}_{L^1_vL^\infty_x\pa{\langle v \rangle^n}} \leq \eta$ then there exists a unique solution $\mathbf{f}$ to the multi-species Boltzmann equation $\eqref{perturbedmultiBE expodecay}$ in $L^1_vL^\infty_x\pa{\langle v \rangle^n}$ with $\mathbf{f_0}$ as initial datum and satisfying the conservation of mass, momentum and energy. Our choice of weight $w$ implies
\begin{equation}\label{control norms}
\norm{\mathbf{f_0}}_{L^1_vL^\infty_x\pa{\langle v \rangle^n}} \leq \pa{\int_{\R^3}\frac{(1+\abs{v}^2)^{n/2}}{w(v)}\:dv} \norm{\mathbf{f_0}}_{L^\infty_{x,v}\pa{\mathbf{w}}}
\end{equation}
and the integral is finite because either $w$ is exponential or $w$ is polynomial of degree $k>k_0\geq 6$ by $\eqref{k0}$ and $n$ belongs to $(2,3)$. Thus $L^\infty_{x,v}\pa{\mathbf{w}} \subset L^1_vL^\infty_x\pa{\langle v \rangle^n}$. Therefore we can use the existence and uniqueness result of \cite[Theorem 2.2]{BriDau} to obtain a unique solution $\mathbf{f}$ to $\eqref{perturbedmultiBE expodecay}$ in $L^1_vL^\infty_x\pa{\langle v \rangle^n}$ if $\norm{\mathbf{f_0}}_{L^\infty_{x,v}\pa{w}}$ is sufficiently small.
\par Moreover, \cite[Section 6]{BriDau} showed existence and uniqueness of $(\mathbf{f_1},\mathbf{f_2})$ such that $(\mathbf{f_1} + \mathbf{f_2})$ satisfies the conservation of mass, momentum and energy. More precisely, \cite[Propositions 6.6 and 6.7]{BriDau} and the control $\eqref{control norms}$ shows that if $\norm{\mathbf{f_0}}_{L^\infty_{x,v}\pa{w}}$ is sufficiently small then there exist a unique $\mathbf{f_1}$ in  $L^1_vL^\infty_x\pa{\langle v \rangle^n}$ solution to $\eqref{f1}$ and a unique solution $\mathbf{f_2}$ of $\eqref{f2}$ in $L^\infty_{x,v}\pa{\langle v \rangle^\beta\boldsymbol\mu^{-1/2}}$ and satisfies the following exponential decay
\begin{equation}\label{expodecay f2}
\exists C_2,\:\lambda_2 >0,\:\forall t\geq 0, \quad \norm{\mathbf{f_2}(t)}_{L^\infty_{x,v}\pa{\langle v \rangle^\beta\boldsymbol\mu^{-1/2}}} \leq C_2e^{-\lambda_2 t}\norm{\mathbf{f_0}}_{L^1_vL^\infty_x\pa{\langle v \rangle^n}}.
\end{equation}
The constants $C_2$ and $\lambda_2$ are constructive.

\bigskip
As a conclusion, the proof of Theorem \ref{theo:main} will follow directly if we can prove that $\mathbf{f_1}$ is indeed in $L^\infty_{x,v}(w)$ and decays exponentially. Thus, it is a consequence of the following proposition.

\bigskip
\begin{prop}\label{prop:expo decay f1}
Let $w_i=e^{\kappa_1(\sqrt{m_i}\abs{v})^{\kappa_2}}$ with $\kappa_1 >0$ and $\kappa_2$ in $(0,2)$ or $w=\langle \sqrt{m_i}v \rangle^k$ with $k> k_0$. Let $\mathbf{f_0}$ be in $L^\infty_{x,v}\pa{\mathbf{w}}$. There exist $\eta_1$, $\lambda_1 >0$ such that if 
$$\norm{\mathbf{f_0}}_{L^\infty_{x,v}\pa{\mathbf{w}}}\leq \eta_1$$
then the function $\mathbf{f_1}$ solution to $\eqref{f1}$ is in $L^\infty_{x,v}\pa{w}$ and satisfies
$$\forall t\geq 0, \quad \norm{\mathbf{f_1}(t)}_{L^\infty_{x,v}\pa{\mathbf{w}}} \leq e^{-\lambda_1 t}\norm{\mathbf{f_0}}_{L^\infty_{x,v}\pa{\mathbf{w}}}.$$
The constants $\eta_1$ and $\lambda_1$ are constructive and only depends on $N$, $\mathbf{w}$ and the collision kernel.
\end{prop}
\bigskip

We adapt a method developed in \cite{Bri6} and improved in \cite{BriGuo} that relies on a Duhamel formula along the characteristic trajectories for $\mathbf{f_1}$.

\bigskip
\begin{proof}[Proof of Proposition \ref{prop:expo decay f1}]
Since $\mathbf{f_1}$ is solution to $\eqref{f1}$ we can write each of its components under their implicit Duhamel form along the characteristics. As the space variable lives on the torus, these trajectories are straight lines with velocity $v$. We thus have the following equality for almost every $(x,v)$ in $\T^3\times\R^3$ and for every $i$ in $\br{1,\dots,N}$.
\begin{equation}\label{f1 Duhamel}
\begin{split}
\forall t\geq 0, \: f_{1i}(t,x,v) =& e^{-\nu_i(v)t}f_{0i}(x-tv,v) + \int_0^t e^{-\nu_i(v)(t-s)}B_i\pa{\mathbf{f_1}}(s,x-(t-s)v,v)\:ds
\\& +\int_0^t e^{-\nu_i(v)(t-s)} Q_i(\mathbf{f_1}+\mathbf{f_2})(s,x-(t-s)v,v)\:ds.
\end{split}
\end{equation}
We would like to bound $w_i(v)f_{1i}(t,x,v)$ in $x$ and $v$ and we therefore bound each of the three terms on the right-hand side.

\bigskip
Since $\nu_i(v) \geq \nu_0 >0$ for all $i$ and $v$ we have the following bound
\begin{equation}\label{1st bound}
\abs{w_i(v)e^{-\nu_i(v)t}f_{0i}(x-tv,v)} \leq e^{-\nu_0 t}\norm{\mathbf{f_0}}_{L^\infty_{x,v}\pa{\mathbf{w}}}.
\end{equation}
\par Multiplying and dividing by $\nu_i(v)$ inside the first integral term on the right-hand side of $\eqref{f1 Duhamel}$ yields
\begin{equation*}
\begin{split}
&\abs{w_i(v)\int_0^t e^{-\nu_i(v)(t-s)}B_i\pa{\mathbf{f_1}}(s,x-(t-s)v,v)\:ds}
\\&\quad\quad\quad\quad\quad\leq \int_0^t \nu_i(v)e^{-\nu_i(v)(t-s)} \norm{B_i\pa{\mathbf{f_1}}(s)}_{L^\infty_{x,v}\pa{w_i\nu_i^{-1}}}\:ds
\\&\quad\quad\quad\quad\quad\leq C_{Bi}\int_0^t \nu_i(v)e^{-\nu_i(v)(t-s)}\norm{\mathbf{f_1}(s)}_{L^\infty_{x,v}\pa{\mathbf{w}}}\:ds,
\end{split}
\end{equation*}
where we used Lemma \ref{lem:control B} to estimate the norm of $B_i\pa{\mathbf{f_1}}$ with $\sum_{i=1}^N C_{Bi} = C_B$. Since for any $\eps$ in $(0,1)$
$$\forall 0\leq s \leq t, \quad -\nu_i(v)(t-s) \leq -\eps\nu_0t -\nu_i(v)(1-\eps)(t-s)+\eps \nu_0 s$$
we can further bound
\begin{equation*}
\begin{split}
&\abs{w_i(v)\int_0^t e^{-\nu_i(v)(t-s)}B_i\pa{\mathbf{f_1}}(s,x-(t-s)v,v)\:ds}
\\&\quad\quad\quad\quad\quad\leq C_{Bi}e^{-\eps\nu_0 t}\pa{\int_0^t \nu_i(v)e^{-\nu_i(v)(1-\eps)(t-s)}\:ds}\sup\limits_{0\leq s\leq t}\cro{e^{\eps\nu_0s}\norm{\mathbf{f_1}(s)}_{L^\infty_{x,v}\pa{\mathbf{w}}}}.
\end{split}
\end{equation*}
We thus conclude for all $\eps$ in $(0,1)$
\begin{equation}\label{2nd bound}
\begin{split}
&\abs{w_i(v)\int_0^t e^{-\nu_i(v)(t-s)}B_i\pa{\mathbf{f_1}}(s,x-(t-s)v,v)\:ds}
\\&\quad\quad\quad\quad\quad \leq \frac{C_{Bi}}{1-\eps}e^{-\eps\nu_0 t} \sup\limits_{0\leq s\leq t}\cro{e^{\eps\nu_0s}\norm{\mathbf{f_1}(s)}_{L^\infty_{x,v}\pa{\mathbf{w}}}}.
\end{split}
\end{equation}
\par Finally, the last integral on the right-hand side of $\eqref{f1 Duhamel}$ is dealt with exactly the same way but using Lemma \ref{lem:control Q} to control the operator $Q_i$. This yields
\begin{equation}\label{3rd bound}
\begin{split}
&\abs{w_i(v)\int_0^t e^{-\nu_i(v)(t-s)}Q_i\pa{\mathbf{f_1}+\mathbf{f_2}}(s,x-(t-s)v,v)\:ds} 
\\&\quad\quad\quad\leq \frac{C_Q}{1-\eps}e^{-\nu_0 t}\pa{\norm{\mathbf{f_1}}_{L^\infty_{[0,t],x,v}(\mathbf{w})}+2\norm{\mathbf{f_2}}_{L^\infty_{[0,t],x,v}(\mathbf{w})}} \sup\limits_{0\leq s\leq t}\cro{e^{\eps\nu_0s}\norm{\mathbf{f_1}(s)}_{L^\infty_{x,v}\pa{\mathbf{w}}}}
\\&\quad\quad\quad\quad + \frac{C_Q}{1-\eps}\sup\limits_{0\leq s\leq t}\cro{e^{\eps\nu_0s}\norm{\mathbf{f_2}(s)}^2_{L^\infty_{x,v}\pa{\mathbf{w}}}}.
\end{split}
\end{equation}

\bigskip
Gathering $\eqref{1st bound}$, $\eqref{2nd bound}$ and $\eqref{3rd bound}$ into $\eqref{f1 Duhamel}$ we see that for all $\eps$ in $(0,1)$,
\begin{equation}\label{finalbound}
\begin{split}
e^{\eps\nu_0 t} \norm{\mathbf{f_1}(t)}_{L^\infty_{x,v}(\mathbf{w})} \leq& \norm{\mathbf{f_0}}_{L^\infty_{x,v}(\mathbf{w})} + \frac{C_Q}{1-\eps}\sup\limits_{0\leq s\leq t}\cro{e^{\eps\nu_0s}\norm{\mathbf{f_2}(s)}^2_{L^\infty_{x,v}\pa{\mathbf{w}}}}
\\&+\sup\limits_{0\leq s\leq t}\cro{e^{\eps\nu_0s}\norm{\mathbf{f_1}(s)}_{L^\infty_{x,v}\pa{\mathbf{w}}}}
\\&\quad\quad\times\cro{\frac{C_B}{1-\eps}+\frac{C_Q}{1-\eps}e^{-\nu_0 t}\pa{\norm{\mathbf{f_1}}_{L^\infty_{[0,t],x,v}(\mathbf{w})}+2\norm{\mathbf{f_2}}_{L^\infty_{[0,t],x,v}(\mathbf{w})}}}.
\end{split}
\end{equation}
We can now conclude by choosing $\eps$ and $\mathbf{f_0}$ sufficiently small. Indeed, first $C_B <1$ (see Remark \ref{rem:CB<1} so we can choose $C_B< 1-\eps$ and thus $\alpha:= 1- C_B/(1-\eps)>0$. Second, the exponential decay of $\mathbf{f_2}$ given by $\eqref{expodecay f2}$ also implies that 
$$e^{\eps\nu_0s}\norm{\mathbf{f_2}(s)}_{L^\infty_{x,v}(\mathbf{w})} \leq e^{\eps\nu_0s}\norm{\mathbf{f_2}(s)}_{L^\infty_{x,v}(\langle v\rangle^\beta\boldsymbol\mu)}\leq C_2\norm{\mathbf{f_0}}_{L^\infty_{x,v}(\mathbf{w})}$$
as long as $\eps \nu_0 \leq \lambda_2$. It thus follows
\begin{equation*}
\begin{split}
\sup\limits_{s\in [0,t]}\cro{e^{\eps\nu_0 s} \norm{\mathbf{f_1}(s)}_{L^\infty_{x,v}(\mathbf{w})}}\leq& C_0 \norm{\mathbf{f_0}}_{L^\infty_{x,v}(\mathbf{w})} + \frac{C_Q}{\alpha(1-\eps)}\pa{\sup\limits_{s\in [0,t]}\cro{e^{\eps\nu_0 s} \norm{\mathbf{f_1}(s)}_{L^\infty_{x,v}(\mathbf{w})}}}^2.
\end{split}
\end{equation*}
Define 
$$T_{\max} = \sup\br{t\geq 0, \:\sup\limits_{s\in [0,t]}\cro{e^{\eps\nu_0 s} \norm{\mathbf{f_1}(s)}_{L^\infty_{x,v}(\mathbf{w})}} < 4C_0 \norm{\mathbf{f_0}}_{L^\infty_{x,v}(\mathbf{w})}},$$
which is well-defined because $t=0$ belongs to the set, and it follows than for any $t$ in $[0,T_{\max}]$,
$$\sup\limits_{s\in [0,t]}\cro{e^{\eps\nu_0 s} \norm{\mathbf{f_1}(s)}_{L^\infty_{x,v}(\mathbf{w})}}\leq C_0\norm{\mathbf{f_0}}_{L^\infty_{x,v}(\mathbf{w})} + 16C_0^2\frac{C_Q}{\alpha(1-\eps)}\norm{\mathbf{f_0}}_{L^\infty_{x,v}(\mathbf{w})}^2$$
which is strictly smaller than $2C_0\norm{\mathbf{f_0}}_{L^\infty_{x,v}(\mathbf{w})}$ if $\norm{\mathbf{f_0}}_{L^\infty_{x,v}(\mathbf{w})}$ is sufficiently small. If $T_{\max}$ is finite then it follows $\sup\limits_{s\in [0,t]}\cro{e^{\eps\nu_0 s} \norm{\mathbf{f_1}(s)}_{L^\infty_{x,v}(\mathbf{w})}} \leq 2C_0\norm{\mathbf{f_0}}_{L^\infty_{x,v}(\mathbf{w})}$ which contradicts the definition of $T_{\max}$. Therefore, $T_{\max}= +\infty$ and
$$\sup\limits_{t\geq 0}\cro{e^{\eps\nu_0 t} \norm{\mathbf{f_1}(t)}_{L^\infty_{x,v}(\mathbf{w})}} \leq 2C_0\norm{\mathbf{f_0}}_{L^\infty_{x,v}(\mathbf{w})},$$
which concludes the proof.
\end{proof}
\bigskip

%% APPENDIX %%%%

\appendix

\section{Proof of Lemma \ref{lem:radiallysym}}\label{appendix:proof lemma radially symmetric}

We will denote by $\delta_d$ the dirac distribution in dimension $d$. First of all, one has the following rewriting (see \cite[equation $(3.6)$]{BobGamPan} or \cite[Lemma 1]{BobRja} or \cite[equation $(4.20)$]{GMM})
\begin{equation}\label{changeofvariable}
\forall \varphi \in C(\R^3),\:\forall w \in \R^3, \quad \int_{\S^2}\varphi(\abs{w}\sigma - w)\:d\sigma = \frac{1}{\abs{w}}\int_{\R^3}\varphi(y)\delta_1\pa{\langle y,w\rangle + \frac{1}{2}\abs{y}^2}\:dy.
\end{equation}
Then, by definition we have
$$v' = v+ m_j(\abs{w}\sigma - w) \quad\mbox{and}\quad v_*' = v_* - m_i(\abs{w}\sigma - w)$$
with $w = (v-v_*)/(m_i+m_j)$. We can therefore use first $\abs{v-v_*} = \abs{v'-v'_*}$ and second $\eqref{changeofvariable}$ to our $Q^+_{ij}$ operator which implies
\begin{eqnarray}
&&\abs{Q^+_{ij}(F,G)(v)} \nonumber
\\ &&\quad\leq b_{ij}^\infty \int_{\R^3} dv_* \int_{\S^2} \abs{(v + m_j(\abs{w}\sigma - w))-(v_* - m_i(\abs{w}\sigma - w))}^\gamma \nonumber
\\&&\quad\quad\quad\quad\quad\quad\quad\quad\quad\quad\quad \abs{F}\pa{v+ m_j(\abs{w}\sigma - w)}\abs{G}\pa{v_* - m_i(\abs{w}\sigma - w)}d\sigma\nonumber
\\&&\quad\leq b_{ij}^\infty(m_i+m_j)\int_{\R^3}\int_{\R^3}\frac{1}{\abs{(v+m_jy)-(v_*-m_iy)}^{1-\gamma}}F(v+m_j y) G(v_* -m_iy)\nonumber
\\&&\quad\quad\quad\quad\quad\quad\quad\quad\quad\quad\quad\times\delta_1\pa{\langle y,w\rangle + \frac{1}{2}\abs{y}^2}dydv_* \nonumber
\\&&\quad\leq b_{ij}^\infty(m_i+m_j)\int_{\R^3}\int_{\R^3}\int_{\R^3}\frac{\abs{F}(v+m_j y) \abs{G}(v_* -m_i z)}{\abs{(v+m_jy)-(v_*-m_iz)}^{1-\gamma}}\nonumber
\\&&\quad\quad\quad\quad\quad\quad\quad\quad\quad\quad\quad\times\delta_1\pa{\langle y,w\rangle + \frac{1}{2}\abs{y}^2}\delta_3\pa{y-z}\:dydzdv_*.\label{Q+ij first bound}
\end{eqnarray}

\bigskip
Defining $v' = v+ m_j y$ and $v'_* = v_*-m_i z$ we compute
$$y-z = \frac{-1}{m_im_j}\pa{m_i v + m_j v_* - m_i v' - m_j v'_*}$$
and also that for $y=z$
$$m_i \abs{v'}^2 + m_j \abs{v'_*}^2 - m_i \abs{v}^2 - m_j \abs{v_*}^2 = 2m_im_j(m_i+m_j)\cro{\langle y,w\rangle + \frac{1}{2}\abs{y}^2}.$$
Denoting
\begin{eqnarray*}
C_m &=& \br{(v,v_*,v',v'_*)\in \pa{\R^3}^4, \quad m_i v + m_j v_* = m_i v' + m_j v'_*}
\\ C_e &=& \br{(v,v_*,v',v'_*)\in \pa{\R^3}^4, \quad m_i \abs{v'}^2 + m_j \abs{v'_*}^2= m_i \abs{v}^2 m_j \abs{v_*}^2},
\end{eqnarray*}
and using the property $\delta_d(ax) = \abs{a}^{-d}\delta_d(x)$ the above implies
$$\delta_1\pa{\langle y,w\rangle + \frac{1}{2}\abs{y}^2}\delta_3\pa{y-z} = 2m_i^4m_j^4(m_i+m_j)\delta_{C_m}\delta_{C_e}$$
where $\delta_A$ is now the distribution on the set $A$. Using the change of variable $(y,z) \to (v',v'_*)$ in $\eqref{Q+ij first bound}$ therefore gives
\begin{equation*}
\abs{Q^+_{ij}(F,G)(v)} \leq 2b_{ij}^\infty m_im_j(m_i+m_j)^2\int_{\R^3}\int_{\R^3}\int_{\R^3}\frac{\abs{F}(v') \abs{G}(v'_*)}{\abs{v'-v'_*}^{1-\gamma}}\delta_{C_m}\delta_{C_e} \:dv_*dv'dv'_*.
\end{equation*}
Since $F$ and $G$ are radially symmetric we infer by using spherical coordinates $v_*=r_*\sigma_*$, $v'=r'\sigma'$ and $v'_* = r'_*\sigma'_*$
\begin{equation}\label{Q+ij second bound}
\begin{split}
\abs{Q^+_{ij}(F,G)(v)} \leq 2b_{ij}^\infty m_im_j(m_i+m_j)^2\int_0^\infty\int_0^\infty\int_0^\infty & \frac{\abs{F}(r') \abs{G}(r'_*)}{\abs{r'-r'_*}^{1-\gamma}}
\\&\delta_{C_e}K(r,r_*,r',r'_*)dr_*dr'dr'_*
\end{split}
\end{equation}
where
$$K(r,r_*,r',r'_*) = (r_*)^2(r')^2(r'_*)^2\int_{\S^2}\int_{\S^2}\int_{\S^2} \delta_{C_m} \:d\sigma_*d\sigma'd\sigma'_*.$$
Note that we used the triangular inequality $\abs{v'-v'_*} \geq \abs{r'-r'_*}$.

\bigskip
From \cite[Step 3 proof of Lemma 4.6]{GMM} we have that for any $(a_1,a_2,a_3,a_4)$ in $\pa{\R^+}^4$
$$A(a_1,a_2,a_3,a_4) = \int_{\S^2}\int_{\S^2}\int_{\S^2} \delta_{\br{a_1\sigma + a_2\sigma_*= a_3\sigma' + a_4\sigma'_*}}\:d\sigma_*d\sigma'd\sigma'_*$$
does not depend on $\sigma$ in $\S^2$ and $A(a_1,a_2,a_3,a_4)$ is invariant under permutations of the variables $(a_1,a_2,a_3,a_4)$. Moreover, \cite{GMM} also proved that for $a_1 > a_2 >a_3>a_4$ one has
$$A(a_1,a_2,a_3,a_4)= \frac{8\pi^2}{a_1a_2a_3a_4}\pa{a_4+a_3 - \max\br{(a_1-a_2), (a_3-a_4)}}.$$
From the latter equality it follows that
$$A(a_1,a_2,a_3,a_4)\leq \frac{16\pi^2}{a_1a_2a_3a_4}a_4 = \frac{16\pi^2}{a_1a_2a_3a_4}\min\br{a_1,a_2,a_3,a_4},$$
which holds for any $(a_1,a_2,a_3,a_4)$ thanks to the permutation invariance. We therefore obtain that
\begin{eqnarray*}
K(r,r_*,r',r'_*) &=& (r_*)^2(r')^2(r'_*)^2 A(m_ir,m_jr_*,m_ir',m_jr'_*)
\\ &\leq& \frac{16\pi^2}{m_i^2m_j^2}\frac{r_*r'r'_*}{r}\min\br{m_ir,m_jr_*,m_ir',m_jr'_*}
\end{eqnarray*}
which we can plug inside $\eqref{Q+ij second bound}$ to obtain
\begin{equation*}
\begin{split}
\abs{Q^+_{ij}(F,G)(v)}\leq 32\pi^2 b_{ij}^\infty \frac{(m_i+m_j)^2}{m_im_j}\int_0^\infty\int_0^\infty\int_0^\infty&\frac{r_*r'r'_*\min\br{m_ir,m_jr_*,m_ir',m_jr'_*}}{r\abs{r'-r'_*}^{1-\gamma}}
\\&\delta_{C_e} \abs{F}(r') \abs{G}(r'_*)dr_*dr'dr'_*
\end{split}
\end{equation*}
which gives the expected result noticing that
\begin{eqnarray*}
\delta_{C_e}\mathbf{1}_{\br{r_*\geq 0}} &=& \delta\pa{m_j\pa{r_*^2 - \pa{\sqrt{\frac{m_i}{m_j}(r')^2 + (r'_*)^2-\frac{m_i}{m_j}r^2}}^2}}\mathbf{1}_{\br{r_*\geq 0}}
\\&=&\frac{1}{2m_jr_*}\delta\pa{r_* -\sqrt{\frac{m_i}{m_j}(r')^2 + (r'_*)^2-\frac{m_i}{m_j}r^2}}\mathbf{1}_{\br{\frac{m_i}{m_j}(r')^2 + (r'_*)^2-\frac{m_i}{m_j}r^2\geq 0}}.
\end{eqnarray*}

%
% Pour une biblio classe
\bibliographystyle{acm}
\bibliography{bibliography_multiBE}

%%%% Pour une biblio manuelle
%%\newpage
%%\include{bibliography}

% On met les signatures
\bigskip
\signmb

\end{document}